\documentclass[aps,twocolumn,english,pra,superscriptaddress]{revtex4-1}
\usepackage[T1]{fontenc}
\usepackage[latin9]{inputenc}
\usepackage{babel}
\usepackage{verbatim}
\usepackage{float}
\usepackage{amsmath}
\usepackage{amsthm}
\usepackage{amssymb}
\usepackage{graphicx}
\usepackage[unicode=true,pdfusetitle,
bookmarks=true,bookmarksnumbered=false,bookmarksopen=false,
breaklinks=false,pdfborder={0 0 1},backref=false,colorlinks=false]
{hyperref}

\usepackage{subfigure,bbm}
\newtheorem{theorem}{Theorem}

\newtheorem{proposition}[theorem]{Proposition}
\newtheorem{remark}[theorem]{Remark}

\makeatother

\newcommand{\ket}[1]{|#1\rangle}
\newcommand{\bra}[1]{\langle #1|}
\newcommand{\be}{\begin{equation}}
\newcommand{\ee}{\end{equation}}
\newcommand{\ba}{\begin{align}}
\newcommand{\ea}{\end{align}}

\DeclareMathOperator{\sech}{sech}

\makeatletter

\providecommand{\tabularnewline}{\\}

\begin{document}

\title{A continuous-variable quantum repeater \\based on quantum scissors and mode multiplexing}

\author{Kaushik P. Seshadreesan}
\email{kaushiksesh@email.arizona.edu}
\affiliation{College of Optical Sciences, University of Arizona, Tucson, AZ 85721}

\author{Hari Krovi}
\affiliation{Quantum Engineering and Computing Physical Sciences and Systems, Raytheon BBN Technologies, Cambridge, MA 02138, USA}

\author{Saikat Guha}
\affiliation{College of Optical Sciences, University of Arizona, Tucson, AZ 85721}

\date{\today}
\begin{abstract}
Quantum repeaters are indispensable for high-rate, long-distance quantum communications. The vision of a future quantum internet strongly hinges on realizing quantum repeaters in practice. Numerous repeaters have been proposed for discrete-variable (DV) single-photon-based quantum communications. Continuous variable (CV) encodings over the quadrature degrees of freedom of the electromagnetic field mode offer an attractive alternative. For example, CV transmission systems are easier to integrate with existing optical telecom systems compared to their DV counterparts. Yet, repeaters for CV have remained elusive. We present a novel quantum repeater scheme for CV entanglement distribution over a lossy bosonic channel that beats the direct transmission exponential rate-loss tradeoff. The scheme involves repeater nodes consisting of a) two-mode squeezed vacuum (TMSV) CV entanglement sources, b) the quantum scissors operation to perform nondeterministic noiseless linear amplification of lossy TMSV states, c) a layer of switched mode multiplexing inspired by second-generation DV repeaters, which is the key ingredient apart from probabilistic entanglement purification that makes DV repeaters work, and d) a non-Gaussian entanglement swap operation. We report our exact results on the rate-loss envelope achieved by the scheme.
\end{abstract}

\maketitle

\section{Introduction}

A {\it quantum internet} \cite{Kimble2008} that distributes entanglement and quantum-secure shared secret key at high rates over large distances exemplifies the goal of quantum communications~\cite{GT07}. It would enable, e.g., unconditionally secure multiparty classical communications~\cite{SPCDLP09}, distributed versions of quantum computation, sensing, and other quantum information processing applications~\cite{PKD18,GJEGF18,ZZS18,RvM16,DDKP07,BR03}. The main hurdle in the way of establishing the quantum internet is photon loss. Entanglement distribution rates over a lossy bosonic channel such as an optical fiber or free space link are known to drop exponentially with loss~\cite{TGW14}. The entanglement distribution capacity of the pure loss bosonic channel of transmissivity $\eta$ under unlimited two-way local operations and classical communication (LOCC) assistance was recently established to be $C_{\textrm{direct}}(\eta)=-\log_2(1-\eta)$ entangled qubit pairs (also called ebits) per mode~\cite{PLOB17, PGBL09} (see also~\cite{WTB17} for a strong converse bound~\footnote{In 2015, Pirandola et al.~\cite{PLOB17} proved $-\log_2(1-\eta)$ ebits/mode as a weak converse upper bound for entanglement distribution over a pure loss channel, which along with the achievability of the same rate proven in~\cite{PGBL09} established it as the capacity. In 2016, Wilde et al.~\cite{WTB17} proved $-\log_2(1-\eta)$ ebits/mode as a strong converse upper bound.}).

{\it Quantum repeaters}~\cite{MATN15,BDCZ98} comprised of entanglement sources, distillation schemes and memories when interspersed over the channel can circumvent this exponential rate-loss tradeoff. For discrete-variable (DV) quantum information encodings such as over quantum states of single photons over their polarization or time-bin degrees of freedom, repeater schemes~\cite{MLKLLJ16,GKFDJC15} based on matter memories \cite{SSMSGR14} as well as optical memories~\cite{MKEG17,ATL15} have been developed. Alternatively, quantum information can also be encoded in the continuous quadrature degrees of freedom of electromagnetic field modes, known as quantum continuous variables (CV). Since CV quantum states reside in infinite dimensional Hilbert spaces, they can hold substantially more quantum information compared to DV states. Also, they can be generated using coherent lasers and nonlinear optics without the need for single-photon detectors, which allows for easier integration with classical telecommunications compared to DV. However, quantum repeaters for CV remain to be well established.

\begin{figure*}[t]
	\begin{tabular}{c}
		\includegraphics[scale=0.2275]{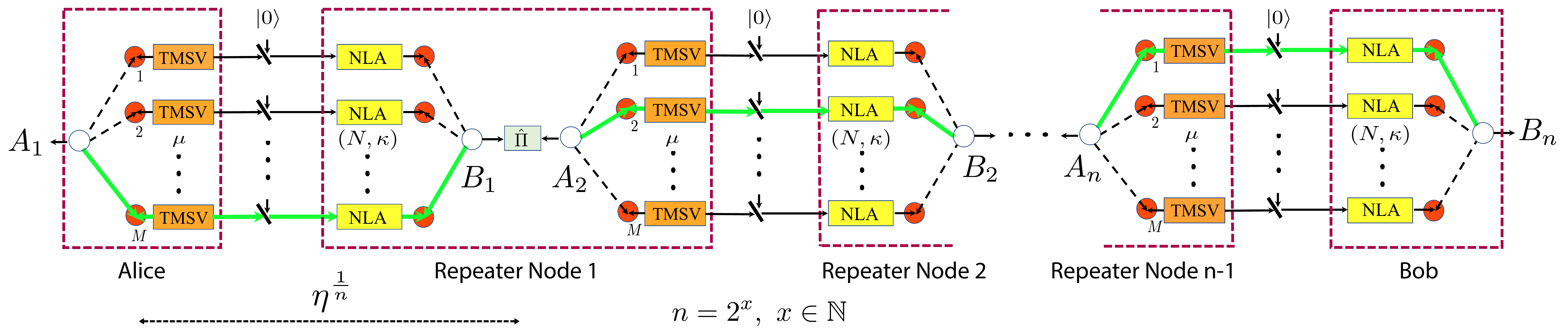}\tabularnewline
	\end{tabular}
	\caption{Mode-multiplexed quantum repeater scheme for CV entanglement distribution based on repeater nodes consisting of two-mode squeezed vacuum (TMSV) state sources, $N-$quantum scissors noiseless linear amplification (NLA), quantum memories (denoted by red circles) and fast optical switches (denoted by white circles with solid, bold, green arrows). The operator $\hat{\Pi}$ refers to a non-Gaussian entangled state projection as described in~(\ref{eq:Bell pair}). The optical switches toggle on to the modes where the NLA operation is successful as indicated in bold (green). $M$ denotes the degree of mode multiplexing.}
	\label{repeater_arch}
\end{figure*}

It is known that for Gaussian CV states, i.e., states with Gaussian quadrature distributions, Gaussian quantum operations, namely physical operations that map Gaussian states to other Gaussian states, alone cannot act as quantum repeaters~\cite{NGGL14,HOvL11}. For the two-mode Gaussian CV state resulting from the action of a pure loss channel on one mode of a two-mode squeezed vacuum (TMSV) entangled state, Ralph proposed a scheme based on non-deterministic noiseless linear amplification (NLA)~\cite{RL09} that probabilistically performs error correction~\cite{Ralph2011}. When the mean photon numbers are small, NLA can be implemented to a good approximation in a heralded fashion by the probabilistic, non-Gaussian quantum scissors operation~\cite{RL09,PPB98}. Dias and Ralph~\cite{DR2017,DR2018} showed that the state heralded upon successful operation of a single quantum scissors on one mode of a lossy TMSV state is more entangled than the lossy TMSV state in terms of the logarithmic negativity~\cite{VW02,Plenio05} and entanglement of formation~\cite{BDSW96} measures. Similarly, the present authors~\cite{SKG181} evaluated the reverse coherent information (RCI)~\cite{GPLS09,PGBL09,DJKR06,DW05,Hcube00} of the state heralded by NLA with multiple quantum scissors on one mode of a lossy TMSV state. The RCI is a lower bound on the distillable entanglement of a state, the latter being the number of ebits that can be distilled from an asymptotically large number of copies of the state using LOCC. It was shown~\cite{SKG181} that the RCI heralded using the (multiple) quantum scissors can exceed $C_{\textrm{direct}}(\eta)$--- a necessary condition for a distillation scheme to be useful in a repeater scheme over the pure loss channel of transmissivity $\eta$. The CV error correction scheme of~\cite{Ralph2011} was recently generalized to the thermal noise channel~\cite{TDR18}. NLA, both ideal~\cite{BLBEGT12}, and approximate, based on the quantum scissors~\cite{GORPR18}, were shown to increase the range of CV quantum key distribution (QKD) over the channel.

In this paper, using repeater nodes consisting of TMSV sources for CV entanglement generation, NLA based on the quantum scissors for entanglement distillation, a layer of switched mode multiplexing, e.g., over spectral or spatial modes, and a non-Gaussian Bell measurement~\cite{FM18} for entanglement swapping, we present a CV quantum repeater scheme (Fig.~\ref{repeater_arch}) that outperforms $C_{\textrm{direct}}(\eta)$. We show that for the proposed scheme NLA based on a single quantum scissors is optimal for entanglement distillation at the repeater nodes compared to any higher number of scissors. This is because the product of the heralded RCI and the heralding success probability at the nodes, when numerically optimized over the free parameters of the system, is found to be maximal for a single quantum scissors (Fig.~\ref{binder1} (a)). We then show that the optimal RCI heralded at the nodes with a single quantum scissors in the limit of infinite NLA gain approaches $1$ independently of the elementary channel segment transmissivity $t$. This implies that the optimal heralded state across a channel segment approaches a perfect ebit, whose RCI by definition is one. The corresponding success probability is found to scale proportional to $t$ when $t\ll1$ (Figs.~\ref{binder1} (b) and (c)). Though the limiting case is unphysical, it carries semblance to DV repeaters, where entanglement distillation is typically based on the successful detection of photons arriving at a repeater node, such that a successful detection event heralds a perfect ebit of entanglement and the detection success probability scales proportional to the transmissivity of the repeater link. This prompts us to consider switched multiplexing over multiple modes (spectral, temporal, spatial, or a combination of any of these) between each pair of adjacent nodes in the proposed CV repeater scheme similar to the so-called second generation DV repeater schemes~\cite{MLKLLJ16}, where mode multiplexing was shown to enable the end-to-end per-mode rates to beat direct transmission~\cite{GKFDJC15, KGDS16}. We show that the rate-loss tradeoff in the DV-like ($\sim1$ ebit/mode) manner of operating the proposed mode-multiplexed CV repeater also similarly beats the direct transmission rate-loss tradeoff (Fig.~\ref{repeater_ratecurve}). However, it is still sub-optimal for the CV repeater. We derive an explicit iterative analytic formula for the end-to-end noisy entangled quantum state heralded across the CV repeater chain indicated in bold, green, in Fig.~\ref{repeater_arch} for any $n_{\textrm{rep}}=2^x-1, \ x\in\mathbb{N}$ number of repeater nodes. Using the exact expression for the end-to-end heralded quantum state, we identify a different operating point in terms of the entanglement source, distillation and swapping parameters, degree of mode multiplexing, and the number of repeater nodes. The said operating point in parameter space results in the individual repeater link states being far from perfect ebits, but with a higher heralding success probability compared to the DV-like mode of operation, thereby resulting in a superior overall end-to-end entanglement distribution rate-loss tradeoff across the repeater chain (Fig.~\ref{real_repeater}).

\begin{figure}[H]
	\centering
	\begin{tabular}{c}
		\includegraphics[scale=0.7]{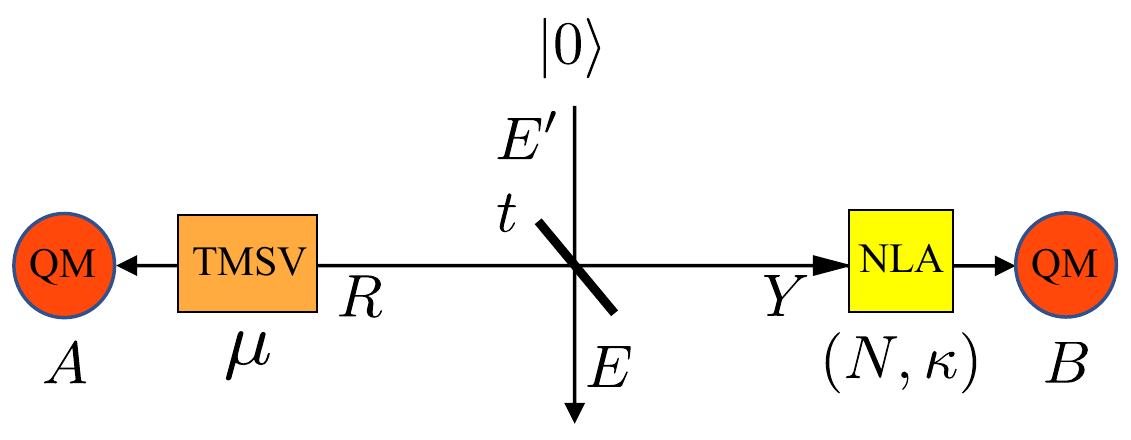}\tabularnewline
	\end{tabular}
	\caption{A single repeater link between adjacent repeater nodes from the CV quantum repeater scheme of Fig.~\ref{repeater_arch}. QM denotes quantum memory.}
	\label{half channel}
\end{figure}

Though our analysis of the proposed CV repeater assumes a pure loss channel model ignoring additional thermal noise encountered in practice, it must be emphasized that our results are a big first step in proving the validity of the concept behind the repeater. Note that in prior work, Furrer and Munro had proposed a CV repeater scheme for the pure loss channel based on alternative non-Gaussian entanglement distillation schemes---symmetric photon replacement and purifying distillation~\cite{Furasek10,BESP03}, which beats $C_{\textrm{direct}}(\eta)$~\cite{FM18}. It is a so-called first generation repeater scheme as per the classification introduced in~\cite{MLKLLJ16} since it involves iterative use of entanglement distillation, which necessitates classical communication between repeater nodes beyond nearest neighbors. On the other hand, our repeater scheme based on mode multiplexing is a second generation scheme that only requires nearest neighbor classical communications. Also, the quantum scissors in comparison to these other distillation schemes involves fewer DV resources, i.e., single photon sources and photon number resolving detectors, making it simpler to implement. Further, while the analysis presented in~\cite{FM18} considers a Gaussified version of the end-to-end heralded non-Gaussian state, our analysis is based on the exact non-Gaussian state heralded across the repeater chain.

This work implicitly assumes that the repeater nodes have access to fast optical switches and multimode quantum memories~\cite{YZHL18,JTLE16,ZHAB15,SSMSGR14,ASDG09} of total effective loss-per-unit-time (inclusive of induced decoherence and coupling losses) less than that of the repeater links connecting adjacent repeater nodes. The quantum scissors and the non-Gaussian entanglement swap operations at the repeater nodes are assumed to be based on ideal single photon sources and photon number resolving detectors.

The paper is organized as follows. Section~\ref{sec:Elementary-Link} presents a detailed analysis of the elementary CV repeater link that constitutes the repeater chain of Fig.~\ref{repeater_arch}. Section~\ref{sec: ngswap} describes the non-Gaussian entanglement swap operation that connects adjacent repeater links in the repeater chain. Section~\ref{sec: MM idea} elucidates the concept behind the mode-multiplexed repeater. Section~\ref{sec: MMCV} contains our main results on the achievable rate-loss tradeoff for the proposed CV repeater scheme based on quantum scissors and mode multiplexing. Section~\ref{sec: disc} concludes the paper with a discussion on questions that are left open in this work and some possible directions for future work.

\section{CV repeater link based on the quantum scissors}
\label{sec:Elementary-Link}
Consider a single repeater link from the proposed CV repeater scheme of Fig.~\ref{repeater_arch}. The link, as shown in Fig. \ref{half channel}, consists of a channel segment of transmissivity, say $t$, a TMSV entangled source of mean photon number $\mu$, NLA of gain $g=\sqrt{(1-\kappa)/\kappa}$ implemented by $N-$quantum scissors (where $\kappa$ is an intrinsic parameter of the scissors), and quantum memories $A$ and $B$. For $N>1,$ the quantum scissors-based NLA module does the following~(c.f. \cite[Fig.~1]{SKG181}): i) Splits the signal quantum state (one share of a lossy TMSV state in this case) into $N$ equal parts. ii) Each subsignal undergoes the quantum scissors operation described in~\cite{PPB98,OMKI01}, which involves linear optics, single photon injection and detection, and as the name suggests truncates the sub-signal quantum state in Fock space to its support on the subspace spanned by the 0 and 1 photon Fock states. (See~\cite{LR94,ISWD97,LK11,LKKLB16} for a related notion of quantum scissors involving nonlinear optical elements). iii) Recombines the "chopped" subsignals into one mode. When the NLA succeeds, it heralds a noiselessly amplified~\cite{PJCC13} version of the signal state that is truncated to its support on the $N-$photon subspace spanned by $0,1,\dots,N$ Fock states. Appendix~\ref{rci_appen} describes the state heralded across the CV repeater link of Fig.~\ref{half channel} in the Fock basis, along with the associated heralding success probability, and an expression for the RCI of the state.

We numerically optimized the {\it true} RCI of the repeater link state, namely the product of the heralded RCI and the heralding success probability, over the TMSV mean photon number and the gain of the quantum scissors, for different number of quantum scissors. The results are plotted in Fig.~\ref{binder1} (a). The channel is assumed to be an optical fiber of attenuation $0.2\ \textrm{dB/km}$, so that the transmissivity of the repeater link as a function of distance $L$ (in km) is $t=10^{-0.02 L}$. Firstly, all the curves lie below $C_\textrm{direct}(t)$, as should be the case by the very definition of capacity. Secondly, the optimal true RCI is the highest for the single quantum scissors. This observation suggests that it is optimal to use NLA based on a single quantum scissors compared to any higher number of scissors for entanglement distillation across the CV repeater link of Fig.~\ref{half channel}.

\begin{figure}
	\includegraphics[scale=0.85]{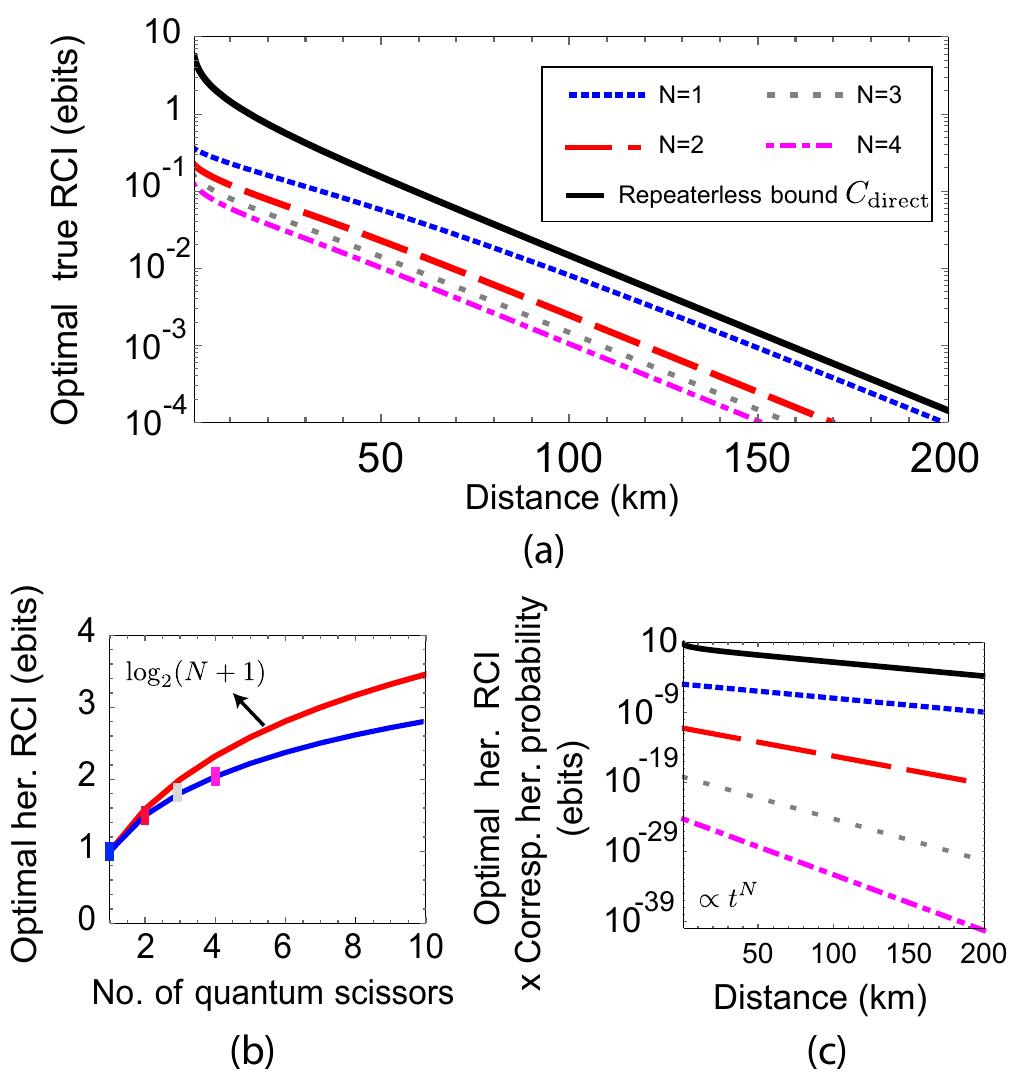}\tabularnewline
	\centering\caption{(a) The optimal true RCI of the repeater link state. (b) The optimal heralded RCI and (c) the corresponding heralding success probability (scaled by the former). The channel is assumed to be an optical fiber of attenuation $0.2\ \textrm{dB/km}$.}
	\label{binder1}
\end{figure}

Further, we also numerically optimized the RCI alone over the same set of parameters as heralded by $N=1,2,3,4$ quantum scissors. The optimal heralded RCI is plotted in Figure~\ref{binder1} (b) as a function of $N$. It is found to approach a limiting constant independent of the channel transmissivity $t$. The constant for a single quantum scissors is found to be $\log_2(2)=1$, which corresponds to the distillation of a perfect ebit of entanglement. For $N>1$, the constant is found to be less than $\log_{2}\left(N+1\right)$, where $N+1$ is the dimensionality of the output Hilbert space, which indicates that the optimal heralded entangled states do not approach perfect ``e-$d$its'' except when $N=1$. In Fig.~\ref{binder1} (c), the asymptotic scaling ($t\ll1$) of the heralding success probability corresponding to the optimal RCI (scaled by the former) is plotted as a function of the transmission distance for different $N$ and found to be $\propto t^{N}$.

Having identified the optimality of NLA based on a single quantum scissors for entanglement distillation (Fig.~\ref{binder1} (a)), we will focus on the CV repeater link consisting of a single quantum scissors for the rest of this paper. It is noteworthy that for this optimal NLA configuration in the CV repeater link in the high loss limit, both the optimal true RCI of  Fig.~\ref{binder1} (a) and the product of the optimal RCI and its corresponding heralding success probability of Fig.~\ref{binder1} (c) scale $\propto t$. However, the former exceeds the latter by several orders of magnitude.


\section{Non-Gaussian entanglement swap}\label{sec: ngswap}
The state heralded across the repeater link of Fig.~\ref{half channel} with a single quantum scissors can be expressed as~\cite{BASRL14}
\begin{align}
|{\psi}\rangle_{ABL} \propto \left(1 +\xi a^\dag b^\dag\right)\sigma_{AL}^{\rho}|{0}\rangle_{ABL}\,\label{bernu1}
\end{align}
where $\hat{a},\ \hat{b}$ are the repeater link mode operators, $L$ is the loss mode, $\xi$ a function of $\mu,\ \kappa,\ t$, and $\sigma_{AL}^{\rho}$ is the two-mode squeezing operator corresponding to squeezing of magnitude $\rho$ in modes $A,\ L$, where $\tanh \rho=\sqrt{1-t}\tanh (\sinh^{-1}\sqrt{\mu})$. See Appendix~\ref{iterform_app} for the exact description with the proportionality constant. Clearly the state in (\ref{bernu1}) is non-Gaussian. In the limit of low TMSV mean photon number, to first approximation, the state in modes $A,\ B$ is a pure state of the form
\begin{equation}
|\psi\rangle_{A_{1}B_{1}}=\left(\left|0\right\rangle _{A_{1}}\left|0\right\rangle _{B_{1}}+\xi\left|1\right\rangle _{A_{1}}\left|1\right\rangle _{B_{1}}\right)/\sqrt{1+\xi^{2}}\label{eq:Bell pair}.
\end{equation} 
At a repeater node, the entanglement in two such repeater link states $|\psi\rangle_{A_{1}B_{1}}$
and $|\psi\rangle_{A_{2}B_{2}}$ can thus be swapped by a non-Gaussian entangled projection operator of the form $\hat{\Pi}=|\phi\rangle\left\langle \phi\right|_{B_{1}A_{2}},$
where $|\phi\rangle_{B_{1}A_{2}}=\left(\left|0\right\rangle _{B_{1}}\left|0\right\rangle _{A_{2}}+q\left|1\right\rangle _{B_{1}}\left|1\right\rangle _{A_{2}}\right)/\sqrt{1+q^{2}},$ with $q=1/\xi$.
Such a projection can be implemented by Fock state filtering~\cite{Furasek10, FM18} and a sequence of displacement operations, photon subtraction and vacuum projection \cite{FM18}. See Appendix~\ref{Appendix: Non-Gaussian-Entanglement-Swap} for details about the implementation and the associated success probability.

\section{Quantum Repeater based on mode multiplexing}\label{sec: MM idea}
In order to describe the idea behind a mode-multiplexed quantum repeater, let us for the moment consider a single-photon-based DV analogue of the CV quantum repeater scheme of Fig.~\ref{repeater_arch}. Let the repeater chain consist of $n=2^x,\ x\in\mathbb{N}$ links (i.e., number of repeater nodes $n_{\textrm{rep}}=2^x-1$) so that the transmissivity of a single link is $t=\eta^{1/n}$. In the DV scheme, entanglement sources at Alice and the repeater nodes generate perfect maximally entangled photon pairs (say polarization Bell pairs), of which one of the photons is transmitted through the lossy channel segment and the other retained in a quantum memory. Successful heralding of the arrival of the transmitted photon at the next node after loss marks the distillation a perfect ebit (RCI $I_R$=1)---an event that happens with a probability $p\propto t=c\eta^{1/n}$. At each repeater node, one local photon that was retained in a quantum memory and one received through the channel are combined on a Bell-basis entangling measurement. The measurement succeeds with a probability $p_{\textrm{swap}}$, accomplishing entanglement swap across the repeater node.

By introducing a layer of mode multiplexing between each pair of adjacent nodes, e.g., using a large number of spectral, temporal, or spatial modes (or a combination of any of these) from the entanglement source, multimode quantum memories, and fast optical switches, the success probability $p$ can be boosted. For $M$ multiplexed modes, the probability that at least one
of them succeeds in distilling an ebit of entanglement is given by
\begin{equation}
p_{M}=1-\left(1-c\eta^{1/n}\right)^{M}.\label{eq:pM_expression}
\end{equation}
For an $n-$link chain, where each link is $M-$mode multiplexed,
the rate at which an ebit of entanglement can be distributed between
the end nodes equals the probability that at least one of the $M$
modes succeeds in each of the $n$ links and the entanglement swaps
at each of the $n-1$ repeater nodes succeeds. It is given by the rate $R$ (in units of ebits/mode) that obeys
\begin{align} 
M\times R & =p_{M}^{n}p_{\textrm{swap}}^{n-1}\leq\begin{cases}
p_{\textrm{swap}}^{n-1}\\
\left(Mc\right)^{n}\eta p_{\textrm{swap}}^{n-1}
\end{cases}\label{eq:multiplexed Rep rate}
\end{align}
From the first upper bound in (\ref{eq:multiplexed Rep rate}), we
have $n=\log\left(M\times p_{\textrm{swap}}\times R_{\textrm{UB}}\right)/\log p_{\textrm{swap}}.$ The two upper bounds intersect
at $\eta=1/\left(Mc\right)^{n}$. From the intersection, we have $n=-\log\eta/\log\left(Mc\right).$
Combining the two, we have 
\begin{align}
&\log\left(M\times p_{\textrm{swap}}\times R_{\textrm{UB}}\right)  =\left(\frac{\log\left(1/p_{\textrm{swap}}\right)}{\log\left(Mc\right)}\right)\log\eta\\
&\Rightarrow  R_{\textrm{UB}}=\frac{1}{M\times p_{\textrm{swap}}}\eta^{\tau},\ \tau=\frac{\log\left(1/p_{\textrm{swap}}\right)}{\log\left(Mc\right)}.\label{eq:Rep s exponent}
\end{align}
For $M>1/(p_{\textrm{swap}}c)$, $\tau<1$. This beats the direct transmission capacity $C_{\textrm{direct}}(\eta)=-\log(1-\eta)$ when $\eta\ll1$ since the latter becomes $\approx1.44\eta$ in the limit. The rate $R_{\textrm{UB}}$ represents an upper bound on the envelope of achievable rates that covers the rates obtainable by varying the number of repeater nodes in the scheme. The exact envelope of achievable rates with the repeater scheme was shown to be~\cite{GKFDJC15}
\begin{align}
R=\frac{1}{M\times p_{\textrm{swap}}}\eta^{s},\ s=\frac{\log\left(p_{\textrm{swap}}\left(1-\left(1-cz\right)^M\right)\right)}{\log z},\label{eq:Rep t exponent}
\end{align}
where $z$ is the unique solution of the transcendental equation
\begin{align}
&\left(1-\left(1-cz\right)^M\right)\log\left(p_{\textrm{swap}}\left(1-\left(1-cz\right)^M\right)\right)\nonumber\\
&=cMz\log z\left(1-cz\right)^{M-1}.\label{eq:zsoln}
\end{align}
Thus, the rate-loss envelope achieved by the mode-multiplexed repeater scheme in the limit of high loss obeys a power law scaling given by $R\propto\eta^s,\ 0<s<1$, and beats $C_{\textrm{direct}}$ which corresponds to $s=1$. The smaller the value of $s$, the greater the range of distances where there is an advantage over direct transmission.

\section{Mode-multiplexed CV repeater}\label{sec: MMCV}
We now present our main results on the entanglement distribution rate-loss tradeoffs achievable with the mode-multiplexed CV repeater scheme of Fig.~\ref{repeater_arch}.

\subsection{DV-like $(\sim1\ \textrm{ebit/mode})$ operation} \label{subsec dv}
Firstly, we consider the case where the quantum scissors at the CV repeater nodes are operated such that the state heralded across a repeater link is a near-perfect ebit with a heralding success probability $\propto t$ (the proportionality constant being $c=5\times10^{-6}$ in the limit $t\ll1$). Since the CV repeater scheme becomes similar to a DV repeater scheme in this case, a direct application of the power-law rate formula of (\ref{eq:Rep t exponent}) along with the numerically optimized value of $p_{\textrm{swap}}=0.00463$ (obtained from (\ref{pphys}) of Appendix~\ref{Appendix: Non-Gaussian-Entanglement-Swap} for $\xi=1$) yields an achievable rate-loss envelope for the mode-multiplexed CV repeater scheme of Fig.~\ref{repeater_arch} as a function of the degree of multiplexing $M$. Figure~\ref{repeater_ratecurve} shows such achievable rate-loss envelopes for $M=10^{10},\ 10^{12}$. The end-to-end channel is once again assumed to be an optical fiber with attenuation $0.2\ \textrm{dB/km}$, i.e., the transmissivity as a function of the communication distance $L$ (in km) is $\eta=10^{-0.02 L}$. The rates are expressed in units of ebits/sec (ebps), where the rate $R$ of (\ref{eq:Rep t exponent}) (in ebits/mode) has been multiplied by a source repetition rate taken to be $R_{\textrm{rep}}=1\ \textrm{MHz}$ modes/sec. With $M\sim 10^{10},$ the rate-loss envelope of the scheme attains a power-law scaling exponent of $s=0.54$, which beats $C_{\textrm{direct}}$ at $851$ km transmission distance. Likewise, $M\sim 10^{12}$ yields an envelope scaling exponent of $s=0.37,$ which beats $C_{\textrm{direct}}$ at $780$ km.

\begin{figure}
	\begin{tabular}{c}
		\includegraphics[scale=0.675]{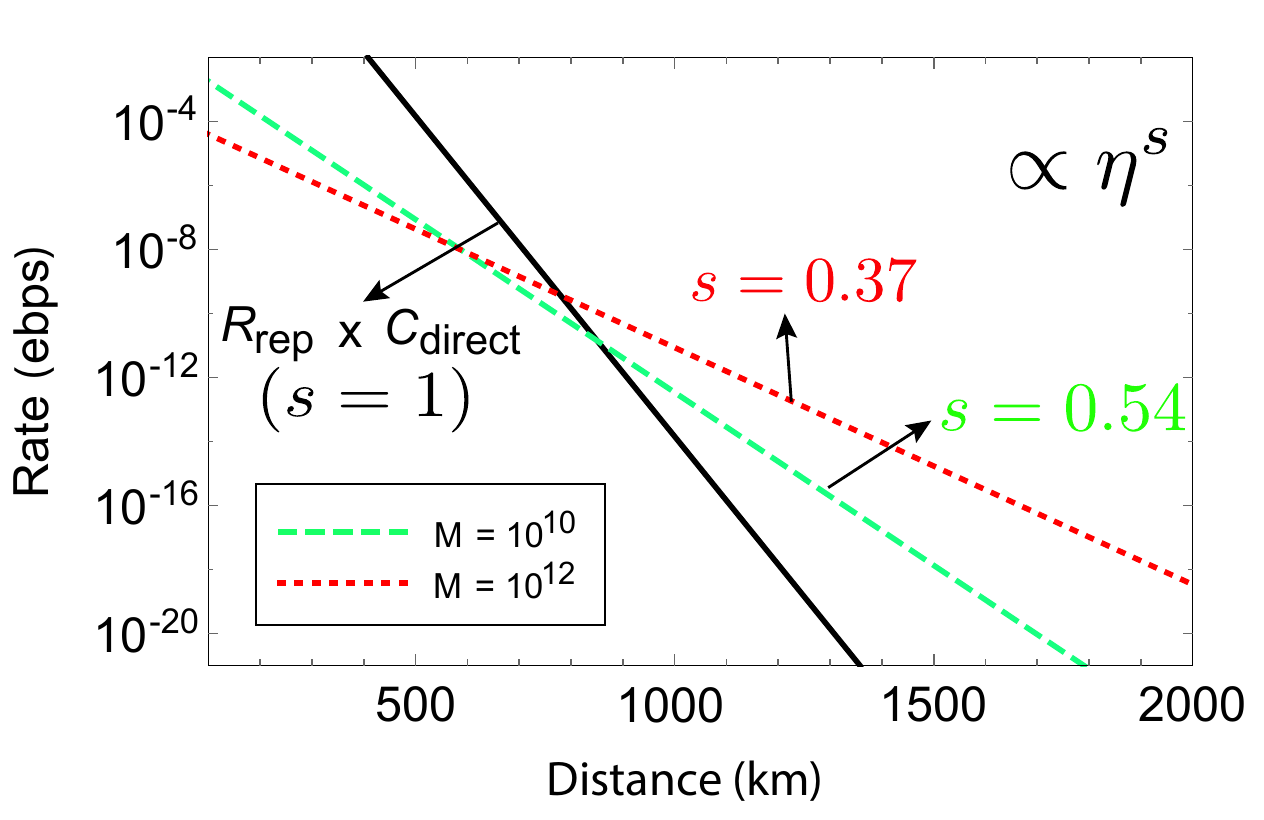}\tabularnewline
	\end{tabular}
	
	\centering
	
	\caption{Envelope of achievable entanglement distribution rates obtained by varying the number of repeater nodes for the DV-like mode of operation of the mode-multiplexed CV repeater scheme presented in Sec.~\ref{subsec dv}. The different $s-$parameter plots correspond to different degrees of multiplexing $M$. The channel is assumed to be an optical fiber of attenuation $0.2\ \textrm{dB/km}$. Source repetition rate $R_{\textrm{rep}}=1\ \textrm{MHz}$.}
	\label{repeater_ratecurve}
\end{figure}

\begin{table}[ht]
	\begin{center}
		\begin{tabular}{cccc}
			\hline
			$ \log_{10}M     $&$\ \ \ \ \ \ s\ \ \ \ \  \ $&$\ \   L_{\textrm{cross}} \textrm{(km)} \ \ $& $R_{\textrm{cross}}\ \textrm{(ebps)}$\\
			\hline\hline
			$<9$ & $-$ & $-$ & $-$\\
			9 & 0.68 & $1066$ & $6.8 \times 10^{-16}$\\
			10 & 0.54 &$851$ & $1.4 \times 10^{-11}$\\
			11 & 0.44 &$788$ & $2.5 \times 10^{-10}$\\						
			12 & 0.37 &$780$ & $3.7 \times 10^{-10}$\\			
			13 & 0.32 &$796$ & $1.7 \times 10^{-10}$\\			
			14 & 0.29 &$833$ & $3.2 \times 10^{-11}$\\			
			15 & 0.26 &$867$ & $6.7 \times 10^{-12}$\\
			16 & 0.23 &$898$ & $1.6 \times 10^{-12}$\\
			\hline
		\end{tabular}
	\end{center}
	\label{tab:DV}
	\caption{Mode-multiplexed CV repeater scheme under the DV-like mode of operation presented in Sec.~\ref{subsec dv} over an optical fiber of attenuation $0.2\ \textrm{dB/km}$: Rate-loss envelope scaling exponent $s$, crossover distance $L_{\textrm{cross}}$ and the corresponding rate $R_{\textrm{cross}}$ at which the repeater-enhanced rate envelope intersects $R_{\textrm{rep}}\times C_{\textrm{direct}}$, as a function of the degree of multiplexing $M$. Source repetition rate $R_{\textrm{rep}}=1\ \textrm{MHz}$.}
\end{table}

For the DV-like mode of operation of the CV repeater, Table~\ref{tab:DV} lists the rate-loss-envelope scaling exponent $s$, and the crossover distance $L_{\textrm{cross}}$ and corresponding rate $R_{\textrm{cross}}$ at which the repeater-enhanced rate-loss envelope intersects $R_{\textrm{rep}}\times C_{\textrm{direct}}$ for different degrees of multiplexing $M$. We draw the following inferences from the table. (i) The value of $M$ has to exceed a threshold ($\sim 10^9$) for the repeater-enhanced rate-loss envelope to beat $R_{\textrm{rep}}\times C_{\textrm{direct}}$. (ii) The exponent $s$, calculated using (\ref{eq:Rep t exponent}) and (\ref{eq:zsoln}), drops monotonically with increasing $M$, which implies greater the value of $M$ greater the range of distances beyond the crossover point where there is an advantage over direct transmission. (iii) The $M-$dependence of $L_{\textrm{cross}}$ and $R_{\textrm{cross}}$ is non-monotonic, and there exists an optimal order of magnitude for $M$ ($\sim 10^{12}$) at which $L_{\textrm{cross}}$ is minimized and $R_{\textrm{cross}}$ maximized.

\begin{figure}
	\begin{tabular}{c}
		\includegraphics[scale=0.8]{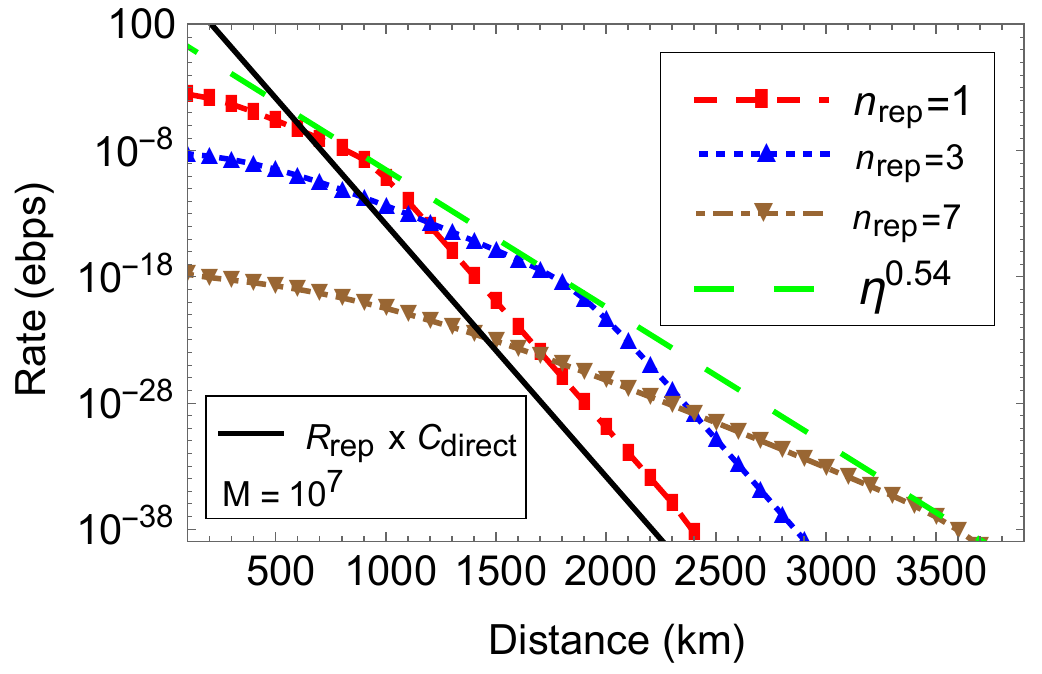}\tabularnewline
	\end{tabular}
	
	\centering
	
	\caption{Entanglement distribution rates achieved by the mode-multiplexed CV repeater scheme under the mode of operation identified in Sec.~\ref{subsec cv}. The value $n_{\textrm{rep}}$ denotes the number of repeater nodes introduced in the channel. The channel is assumed to be an optical fiber of attenuation $0.2\ \textrm{dB/km}$. Source repetition rate $R_{\textrm{rep}}=1\ \textrm{MHz}$.}
	\label{real_repeater}
\end{figure}

\begin{table}[ht]
	\begin{center}
		\begin{tabular}{cccc}
			\hline
			$ \log_{10}M     $&$\ \ \ \ \ \ s\ \ \ \ \  \ $&$\ \ \ \ L_{\textrm{cross}}\ \textrm{(km)}\ \ \ \ $& $R_{\textrm{cross}}\ \textrm{(ebps)}$\\
			\hline\hline
			$<3$ & $-$ & $-$ &$-$\\
			3 & 0.94 & $685$ & $2.5 \times 10^{-8}$\\
			3.4 & 0.87 &$470$ & $5.6 \times 10^{-4}$\\
			3.7 & 0.81 &$400$ & $1.4 \times 10^{-2}$\\						
			4 & 0.76 &$380$ & $3.6 \times 10^{-2}$\\			
			5 & 0.65 &$400$ & $1.2 \times 10^{-2}$\\			
			6 & 0.58 &$455$ & $1.1 \times 10^{-3}$\\			
			7 & 0.54 &$525$ & $4.5 \times 10^{-5}$\\
			8 & 0.52 &$606$ & $1.1 \times 10^{-6}$\\
			9 & 0.5 &$685$ & $3.2 \times 10^{-8}$\\
			10 & 0.49 &$765$ & $6.5 \times 10^{-10}$\\
			\hline
		\end{tabular}
	\end{center}
	\label{tab:CV1}
	\caption{Mode-multiplexed CV repeater scheme under the mode of operation identified in Sec.~\ref{subsec cv} over an optical fiber of attenuation $0.2\ \textrm{dB/km}$: Rate-loss envelope scaling exponent $s$, crossover distance $L_{\textrm{cross}}$ and the corresponding rate $R_{\textrm{cross}}$ at which the repeater-enhanced rate envelope intersects $R_{\textrm{rep}}\times C_{\textrm{direct}}$, as a function of the degree of multiplexing $M$. Source repetition rate $R_{\textrm{rep}}=1\ \textrm{MHz}$.}
\end{table}

\subsection{General operation}\label{subsec cv} 
While the DV-like operation of the CV repeater allows for easy characterization of its achievable rate-loss envelope, it is not optimal. We explore other operating points in parameter space for the repeater scheme by writing down the end-to-end heralded state explicitly for $n=2^x,\ x\in\mathbb{N}$ number of repeater links in the repeater chain, denoted as $\hat{\rho}_{A_1B_{n}}$, and evaluating its RCI. See Appendix~\ref{iterform_app} for an iterative formula for the end-to-end heralded state and its RCI. The rate per mode of a $M-$mode multiplexed repeater chain $R$ is given by
\begin{align}
R=\frac{I_R (\hat{\rho}_{A_1B_{n}})\times \left(1-(1-p_{\textrm{sciss}})^M\right)^{n-1}\times p_{\textrm{swap}}^{n-1}}{M},
\end{align}
where $p_{\textrm{sciss}}$ is the heralding success probability of the quantum scissors in a repeater link and $p_{\textrm{swap}}$ is the entanglement swap success probability associated with the non-Gaussian entangled state projection of Section~\ref{sec: ngswap} (inclusive of its physical implementation) for connecting two repeater links. We identified an operating point consisting of TMSV mean photon number of $0.0719$, quantum scissors gain governed by a power law given by $\kappa=k (\eta^{1/n})^u,$ where $k=0.0557$ and $u=0.6057$, and entanglement swap parameter $q=1/\xi$ in relation to (\ref{bernu1}). Fig.~\ref{real_repeater} shows the rate-loss tradeoff curves corresponding to this new mode of operation of the CV repeater scheme for a degree of multiplexing $M=10^7$ and different number of repeater nodes $n_{\textrm{rep}}=1,3,7$, along with the rate-loss envelope that tangentially meets the said rate curves. The source repetition rate is taken to be $R_{\textrm{rep}}=1\ \textrm{MHz}$. The envelope is found to scale as $\propto\eta^{0.54}$ and beats $R_\textrm{rep} \times C_\textrm{direct}$ at a distance of $~525$ km.

Table~\ref{tab:CV1} lists the rate-loss-envelope scaling exponent $s$, the crossover distance $L_{\textrm{cross}}$ and the corresponding rate $R_{\textrm{cross}}$ for the new mode of operation of  the CV repeater for different degrees of multiplexing $M$. The trends are similar to the DV-like operation. The value of $M$ has to similarly exceed a threshold for the repeater-enhanced rate-loss tradeoff to beat $R_{\textrm{rep}}\times C_{\textrm{direct}}$. The threshold $M$ in the new mode of operation, however, is smaller ($\sim 10^3$) compared to that of the DV-like operation ($\sim 10^9$). The exponent $s$ again drops monotonically with increasing $M$. However, the $M$ required to attain a given $s$ is smaller. For example, whereas the DV-like operation required $M=10^{10}$ to attain $s=0.54$, the new mode of operation attains the same $s$ with $M=10^7$. The $M-$dependence of $L_{\textrm{cross}}$ and $R_{\textrm{cross}}$ is again similarly non-monotonic. The optimal $M$, however, is smaller ($\sim 10^{4}$) compared to that of the DV-like operation ($\sim 10^{12}$), and yields a smaller crossover distance of $L_{\textrm{cross}}=380$ km and higher rate $R_{\textrm{cross}}=3.6\times 10^{-2}$ compared to $L_{\textrm{cross}}=780$ km and $R_{\textrm{cross}}=3.7\times 10^{-10}$ of the DV-like operation.

We now take a closer look at the rate-loss envelopes corresponding to the same scaling exponent $s=0.54$ in the two modes of operation as in Figs.~\ref{repeater_ratecurve} and \ref{real_repeater}. As pointed out earlier, the new mode of operation requires fewer number of multiplexed modes ($M=10^7$) compared to DV-like operation ($10^{10}$) to attain this scaling exponent. The crossover distance $L_{\textrm{cross}}=525$ km is smaller and the corresponding rate $R_{\textrm{cross}}=4.5\times 10^{-5}$ higher compared to the DV-like one, for which $L_{\textrm{cross}}=851$ km and $R_{\textrm{cross}}=1.4\times 10^{-11}$. Also, at any given distance, the new mode of operation results in a higher rate compared to the DV-like operation. For example, at a distance of 1700 km, the former attains a rate $R\sim10^{-17}$~ebps (with 3 repeater nodes), whereas the latter attains $R\sim10^{-20}$~ebps. Together, the above results clearly confirm that the new mode of operation of the mode-multiplexed CV repeater identified in Sec.~\ref{subsec cv} is significantly better than the DV-like mode of operation of Sec.~\ref{subsec dv}, and showcases the true potential of CV.

\section{Discussion and Outlook}\label{sec: disc}
Our results evidently demonstrate that the proposed quantum repeater scheme for CV entanglement distribution \textit{in-principle} works and beats direct transmission. It is important to emphasize that the entanglement source repetition rate of $R_{\textrm{rep}}=1\ \textrm{MHz}$ was chosen as such to ensure that the corresponding requirements on the multimode quantum memories used at the repeater nodes are met under current technologies. For example, the comb preparation and memory read/write times of a multimode quantum memory based on atomic frequency comb generated from rare-earth-ion-doped crystals are typically of the order of microseconds~\cite{YZHL18}. Thus, at the moment direct transmission systems can be operated at much higher repetition rates compared to the proposed CV-repeater scheme, thereby achieving higher entanglement distribution rates in ebps units. However, when faster multimode quantum memories become available in the future, the CV-repeater-enhanced entanglement distribution rates (ebps) reported in this work can be improved commensurately by choosing similarly high source repetition rates, thereby restoring the advantage promised by the repeater scheme.

Regarding future work, since we assumed a pure loss channel model, and ideal single photon sources, photon number resolving detectors, quantum memories and optical switches, the impact of excess thermal noise in the channel and imperfections in these elements on the performance of the scheme remains to be investigated. For example, with regard to the quantum memories, the length of the CV repeater links for the mode of operation identified in Sec.~\ref{subsec cv} is found to be $\sim400\ \textrm{km}$, which necessitates memory storage times of $\sim1.3 \ \textrm{ms}$. Though there is hope to attain longer storage times in the future~\cite{ZHAB15}, the current state-of-the-art storage time for multimode quantum memories remains to be $\sim50\mu \textrm{s}$~\cite{JTLE16}. Beyond this time, the memory would begin to decohere, which has to be taken into account. Within the pure loss channel model, our result identifies a particular operating point for the proposed mode-multiplexed CV repeater scheme in parameter space where the scheme beats direct transmission. The question of what is the optimal performance of the proposed repeater scheme remains open. Another important general question is regarding how the rates supported by CV repeaters compare with those supported by DV repeaters. At the outset the rates achieved by CV repeaters, both in this work and in~\cite{FM18}, seem lower compared to the DV repeater rates reported in the literature. However, upon closer look, when the ideal single photon sources involved in both CV and DV repeaters are replaced by heralded single photon generation from TMSV sources, the normalized entanglement distribution rates per use of a TMSV source is higher for CV compared to DV at large distances~\cite{FM18}. For a careful comparison of CV vs DV repeaters, c.f.~\cite{DMRN19}.

Some possible new directions for future work on CV repeaters include developing CV analogues of all-optical schemes based on cluster state quantum memories~\cite{MKEG17,ATL15}, the use of sources that generate hybrid entanglement such as Bell states in the Gottesman-Kitaev-Preskill~\cite{GKP01}-encoded qubit basis, and considering alternative repeater architectures such as the notion of a third generation, one-way repeater scheme based on quantum error correction, logical Bell state measurements and teleportation~\cite{MLKLLJ16}, which could potentially increase the rates further. Also, CV repeaters for more general network scenarios involving multiple communicating parties largely remains to be explored. Such work might pave the way towards bridging the gap between achievable entanglement distribution rates in repeater networks and the corresponding repeater-assisted end-to-end rate capacities~\cite{Pir19} (see also~\cite{Pir16}).


\begin{acknowledgments}
This work was supported by the Office of Naval Research program Communications and Networking with Quantum Operationally-Secure Technology for Maritime
Deployment (CONQUEST), awarded under Raytheon BBN Technologies prime contract number N00014-16-C-2069, and a subcontract to University
of Arizona. This document does not contain technology or technical data controlled under either the U.S. International Traffic in Arms Regulations or the U.S. Export Administration Regulations. KPS thanks Frederic Grosshans and William Munro for helpful discussions.
\end{acknowledgments}

\appendix
\section{A Lower Bound on the Distillable Entanglement of the Repeater Link State}\label{rci_appen}

In this Appendix, we write down the state heralded across the continuous variable (CV) repeater link in Fig. 1 of the main text for any $N\geq1$ number of quantum scissors and derive its reverse coherent information (RCI)~\cite{DW05,DJKR06,GPLS09,PGBL09}. The RCI is a proven information theoretic lower bound on a state's distillable entanglement in the asymptotic limit of a large number of copies of the state. 

Consider that the two-mode squeezed vacuum (TMSV) state can be expressed
in the Fock basis as
\begin{equation}
\left|\psi\right\rangle _{AR}=\sqrt{{1-\chi^{2}}}\sum_{n=0}^{\infty}\chi^{n}\left|n\right\rangle _{A}\left|n\right\rangle _{R},\label{eq:TMSV Fock}
\end{equation}
where $\chi=\tanh\left(\sinh^{-1}\sqrt{\mu}\right),$ $\mu$ being the mean photon number in each mode. By modeling the pure loss channel of transmissivity $t$ with a beam splitter of the same transmissivity acting on the signal mode $R$ and the environment mode $E$ which is in the vacuum state, we obtain a three-mode output state ($R\rightarrow Y$) of the form
\begin{align}
&\left|\psi\right\rangle _{AYE}/\sqrt{{1-\chi^{2}}}\nonumber\\
& =\sum_{n=0}^{\infty}\chi^{n}\sum_{k=0}^{n}\sqrt{{n \choose k}}x^{k}y^{n-k}\left|n\right\rangle _{A}\left|n-k\right\rangle _{Y}\left|k\right\rangle _{E}\label{eq:AYE state}\\
& =\sum_{k=0}^{\infty}\sum_{n=k}^{\infty}\chi^{n}\sqrt{{n \choose k}}x^{k}y^{n-k}\left|n\right\rangle _{A}\left|n-k\right\rangle _{Y}\left|k\right\rangle _{E}\label{eq:AYE2 state}\\
& =\sum_{k=0}^{\infty}\sum_{m=0}^{\infty}\chi^{n}\sqrt{{m+k \choose k}}x^{k}y^{m}\left|m+k\right\rangle _{A}\left|m\right\rangle _{Y}\left|k\right\rangle _{E},\label{eq:AYE3 state}
\end{align}
where $x=\sqrt{1-t}$ and $y=\sqrt{t}$.

When NLA is successfully applied on the mode $Y$ ($Y\rightarrow B$) using $N-$quantum scissors~(c.f. \cite[Fig. 1]{SKG181} for a schematic of the NLA), the state heralded across Alice, Bob and the environment, and the heralding success probability are given by
\begin{align}
\left|\psi\right\rangle _{ABE} &=\frac{c}{\sqrt{P_{N}}}\sum_{k=0}^{\infty}a^{k}\sum_{m=0}^{N}\frac{N!}{\left(N-m\right)!}\sqrt{{m+k \choose k}}b^{m} \nonumber\\
&\times\left|m+k\right\rangle _{A}\left|m\right\rangle _{B}\left|k\right\rangle _{E},\label{eq:ABE state}\\
P_{N}  &=c^{2}\sum_{k}a^{2k}\sum_{m=0}^{N}\left(\frac{N!}{\left(N-m\right)!}\right)^{2}{m+k \choose k}b^{2m},\label{eq:P_succ_HC}
\end{align}
where $a=\chi\sqrt{1-t},\,b=g\chi\sqrt{t}/N,$ $c=\sqrt{\left(1-\chi^{2}\right)\kappa^{N}},$
$g$ is the NLA gain of the quantum scissors, and $\kappa=1/\left(1+g^{2}\right)$
is an intrinsic parameter in the quantum scissors.

The final two-mode state heralded across the NLA is obtained by tracing
over the loss mode $E$ as $\rho_{AB}^{\left(N\right)}=\sum_{u=0}^{\infty}\rho_{AB}^{\left(N\right)}\left(u\right),$
where 
\begin{equation}
\rho_{AB}^{\left(N\right)}\left(u\right)=\sum_{m=0}^{N}\sum_{m'=0}^{N}\zeta_{m,u}^{\left(N\right)}\zeta_{m',u}^{\left(N\right)}\left|m+u,m\right\rangle \left\langle m'+u,m'\right|_{AB},\label{eq:rho_u}
\end{equation}
and the coefficients $\zeta_{m,u}$ are given by
\begin{align}
\zeta_{m,u}^{\left(N\right)} & =ca^{u}b^{m}\frac{N!}{\left(N-m\right)!}\sqrt{{m+u \choose u}}.
\end{align}
The state $\rho_{AB}^{\left(N\right)}$ is thus
\begin{align}
\rho_{AB}^{\left(N\right)} & =\sum_{u=0}^{\infty}\sum_{i=0}^{N}\left(\zeta_{i,u}^{\left(N\right)}\right)^{2}\left|\Phi_{N}\left(u\right)\right\rangle \left\langle \Phi_{N}\left(u\right)\right|_{AB}\label{eq: diag rho AB}\\
\left|\Phi_{N}\left(u\right)\right\rangle  & _{AB}=\frac{\sum_{i=0}^{N}\left(\frac{\zeta_{i,u}^{\left(N\right)}}{\zeta_{N,u}^{\left(N\right)}}\right)\left|u+i,i\right\rangle _{AB}}{\sqrt{\sum_{i=0}^{N}\left(\frac{\zeta_{i,u}^{\left(N\right)}}{\zeta_{N,u}^{\left(N\right)}}\right)^{2}}},
\end{align}
so that its entropy is given by 
\begin{equation}
H\left(AB\right)=-\sum_{u=0}^{\infty}\left(\sum_{i=0}^{N}\left(\zeta_{i,u}^{\left(N\right)}\right)^{2}\right)\log_{2}\left(\sum_{i=0}^{N}\left(\zeta_{i,u}^{\left(N\right)}\right)^{2}\right).\label{eq: entropy AB}
\end{equation}

The state on system $A$ is obtained by tracing over $B$ as $\rho_{A}^{\left(N\right)}=\operatorname{Tr}_{B}\left(\rho_{AB}^{\left(N\right)}\right)$

\begin{align}
\rho_{A}^{\left(N\right)} & =\sum_{u=0}^{\infty}\sum_{m=0}^{N}\left(\zeta_{m,u}^{\left(N\right)}\right)^{2}\left|m+u\right\rangle \left\langle m+u\right|_{A}\label{eq: rho A}\\
& =\left(\sum_{u=0}^{N}\Gamma_{1}\left(u\right)+\sum_{u=N+1}^{\infty}\Gamma_{2}\left(u\right)\right)\left|u\right\rangle \left\langle u\right|_{A}.\label{eq: rho A1}
\end{align}
where 
\begin{align}
\Gamma_{1}\left(u\right)=\sum_{\begin{array}{c}
	\left\{ i,j\right\} \geq0,\\
	i+j=u
	\end{array}}\left(\zeta_{i,j}^{\left(N\right)}\right)^{2}, & \:\Gamma_{2}\left(u\right)=\sum_{i=0}^{N}\left(\zeta_{i,u-i}^{\left(N\right)}\right)^{2}.
\end{align}
Its entropy is therefore given by
\begin{align}
H\left(A\right) & =-\sum_{u=0}^{N}\Gamma_{1}\left(u\right)\log_{2}\Gamma_{1}\left(u\right)-\sum_{u=N+1}^{\infty}\Gamma_{2}\left(u\right)\log_{2}\Gamma_{2}\left(u\right).\label{eq: Entropy A-1}
\end{align}
Thus, the RCI of the heralded state follows from (\ref{eq: entropy AB})
and (\ref{eq: Entropy A-1}) as
\begin{align}
I_{R}=H\left(A\right)-H\left(AB\right).\label{rci_def}
\end{align}

\section{Non-Gaussian Entanglement Swap\label{Appendix: Non-Gaussian-Entanglement-Swap}}

A beam splitter of transmissivity $T=\cos^{2}\theta$ acting on modes
$i$ and $j$ can be described by the unitary operator:
\begin{equation}
U_{ij}\left(\theta\right)=\exp\left(-\theta\left(\hat{a}_{i}^{\dagger}\hat{a}_{j}-\hat{a}_{j}^{\dagger}\hat{a}_{i}\right)\right),
\end{equation}
that transforms the mode operators as
\begin{align}
\left(\begin{array}{c}
\hat{a}_{i}\\
\hat{a}_{j}
\end{array}\right) & \rightarrow U_{ij}\left(\theta\right)^{\dagger}\left(\begin{array}{c}
\hat{a}_{i}\\
\hat{a}_{j}
\end{array}\right)U_{ij}\left(\theta\right)\\
& =\left(\begin{array}{cc}
\cos\theta & \sin\theta\\
-\sin\theta & \cos\theta
\end{array}\right)\left(\begin{array}{c}
\hat{a}_{i}\\
\hat{a}_{j}
\end{array}\right).\label{eq:BSinputoutput}
\end{align}
It's action on a Fock state input $\left|n\right\rangle _{i}\otimes\left|0\right\rangle _{j}$
is given by
\begin{align}
&U_{ij}\left(\theta\right)\left|n\right\rangle _{i}\left|0\right\rangle _{j}\nonumber\\
&=\sum_{k=0}^{n}\sqrt{{n \choose k}}\left(\cos^{2}\theta\right)^{k/2}\left(\sin^{2}\theta\right)^{\left(n-k\right)/2}\left|k\right\rangle _{i}\left|n-k\right\rangle _{j}.\label{eq:BSFock}
\end{align}

\begin{figure}
	\begin{tabular}{c}
		\includegraphics[scale=0.3]{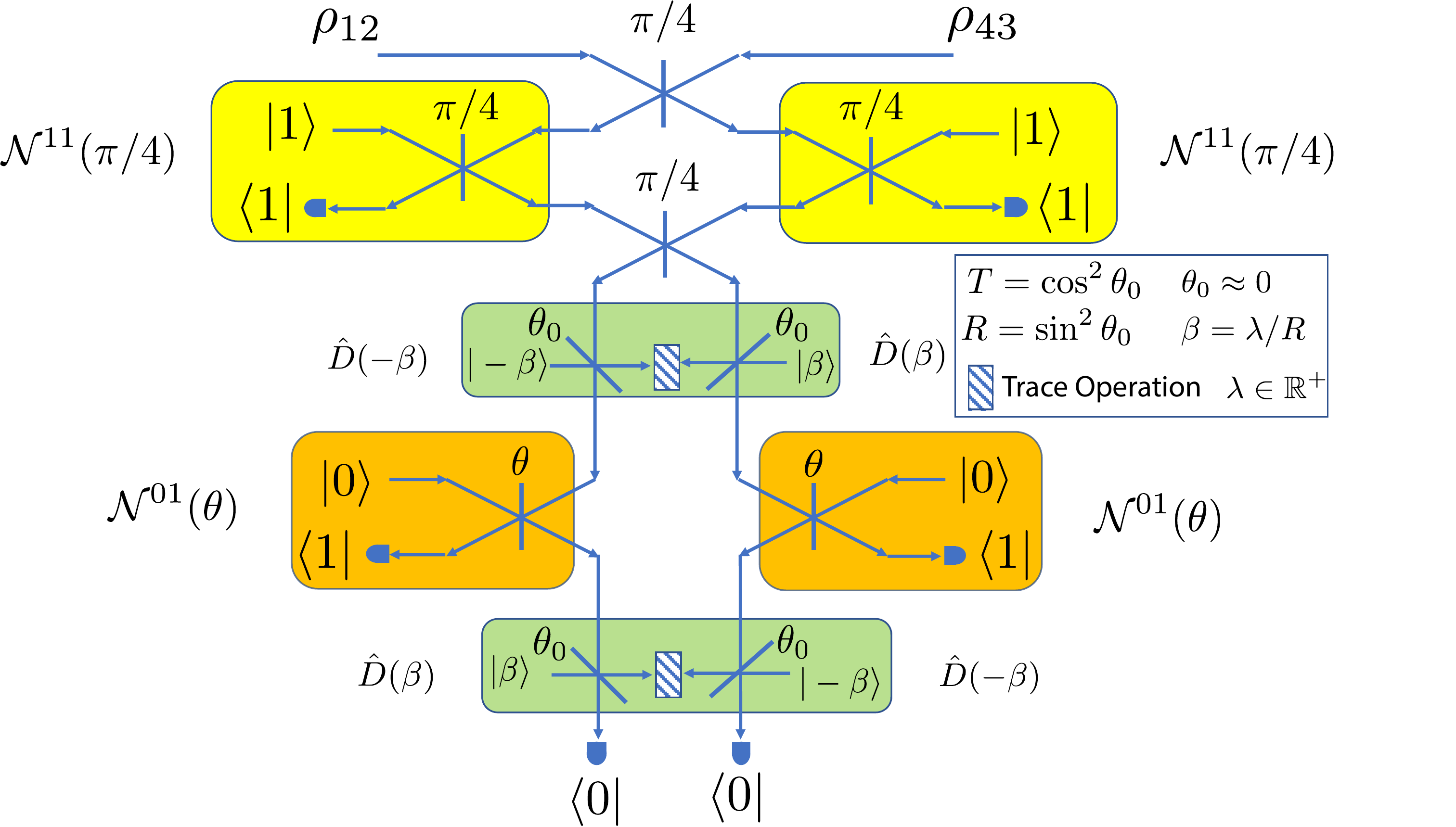}\tabularnewline
	\end{tabular}
	
	\centering
	
	\caption{Non-Gaussian entanglement swap operation based on Fock state filtering, displacement operations, photon subtraction and vacuum projection.
		\label{Munroproj}}
\end{figure}

\begin{proposition}
	\label{prop:N11}Consider the channel composed of mixing an input
	mode $i$ with a mode $j$ in the Fock state $\left|1\right\rangle _{j}$
	on a beam splitter $U_{ij}\left(\theta\right)$, followed by a projective
	measurement $\left\langle 1\right|_{i}$ in mode $i.$ Let us denote
	this non-trace-preserving map as $\mathcal{N}_{i\rightarrow j}^{11}\left(\theta\right)$.
	It can be written as
	\begin{equation}
	\mathcal{N}_{i\rightarrow j}^{11}\left(\theta\right)=\left(-\sin\theta+\cos\theta\frac{d}{d\theta}\right)\left(\sin\theta\right)^{\hat{n}_{ij}},\label{eq:N11channel}
	\end{equation}
	where we use the notation $\hat{n}_{ij}\left|n\right\rangle _{i}=\left|n\right\rangle _{j}$
	to denote input to output transformation.
\end{proposition}

\begin{proof}
	We have
	\begin{align}
	U_{ij}\left(\theta\right)\left|n\right\rangle _{i}\left|1\right\rangle _{j} & =U_{ij}\left(\theta\right)\hat{a}_{j}^{\dagger}\left|n\right\rangle _{i}\left|0\right\rangle _{j}\\
	& =\left(-\sin\theta\hat{a}_{i}^{\dagger}+\cos\theta\hat{a}_{j}^{\dagger}\right)U_{ij}\left(\theta\right)\left|n\right\rangle _{i}\left|0\right\rangle _{j}
	\end{align}
	This implies
	\begin{widetext}
	\begin{align}
	\left\langle 1\right|_{i}U_{ij}\left(\theta\right)\left|n\right\rangle _{i}\left|1\right\rangle _{j} &=\left\langle 0\right|_{i}\hat{a}_{i}\left(-\sin\theta\hat{a}_{i}^{\dagger}+\cos\theta\hat{a}_{j}^{\dagger}\right)U_{ij}\left(\theta\right)\left|n\right\rangle _{i}\left|0\right\rangle _{j}\\
	& =-\sin\theta\left\langle 0\right|_{i}\hat{a}_{i}\hat{a}_{i}^{\dagger}U_{ij}\left(\theta\right)\left|n\right\rangle _{i}\left|0\right\rangle _{j}+\cos\theta\left\langle 0\right|_{i}\hat{a}_{i}\hat{a}_{j}^{\dagger}U_{ij}\left(\theta\right)\left|n\right\rangle _{i}\left|0\right\rangle _{j}\\
	& =-\sin\theta\left\langle 0\right|_{i}\left(1+\hat{a}_{i}^{\dagger}\hat{a}_{i}\right)U_{ij}\left(\theta\right)\left|n\right\rangle _{i}\left|0\right\rangle _{j}+\cos\theta\left\langle 0\right|_{i}\hat{a}_{i}\hat{a}_{j}^{\dagger}U_{ij}\left(\theta\right)\left|n\right\rangle _{i}\left|0\right\rangle _{j}\\
	& =-\sin\theta\left\langle 0\right|_{i}U_{ij}\left(\theta\right)\left|n\right\rangle _{i}\left|0\right\rangle _{j}+\cos\theta\left\langle 0\right|_{i}\hat{a}_{i}\hat{a}_{j}^{\dagger}U_{ij}\left(\theta\right)\left|n\right\rangle _{i}\left|0\right\rangle _{j}
	\end{align}
	\begin{align}
	& =-\sin^{n+1}\theta\left|n\right\rangle _{j}+n\cos^{2}\theta\sin^{n-1}\theta\left|n\right\rangle _{j}\\
	& =\left(-\sin\theta+\cos\theta\frac{d}{d\theta}\right)\sin^{n}\theta\left|n\right\rangle _{j},
	\end{align}
	\end{widetext}
	which implies, for a general state $\left|\psi\right\rangle _{i}=\sum_{n}c_{n}\left|n\right\rangle _{i},$
	the action of the channel is as given in (\ref{eq:N11channel}).
\end{proof}

\begin{proposition}
	\label{prop:N01}Consider the channel composed of mixing an input
	mode $i$ with a mode $j$ in the Fock state $\left|0\right\rangle _{j}$
	on a beam splitter $U_{ij}\left(\theta\right)$, followed by a projective
	measurement $\left\langle 1\right|_{i}$ in the mode $i.$ Let us
	denote this non-trace-preserving map as $\mathcal{N}_{i\rightarrow j}^{01}\left(\theta\right)$.
	It can be written as
	\begin{equation}
	\mathcal{N}_{i\rightarrow j}^{01}\left(\theta\right)=\cos\theta\left(\sin\theta\right)^{\hat{n}_{ij}}\hat{a}_{j}\label{eq:N01channel}
	\end{equation}
	where we use the notation $\hat{n}_{ij}\left|n\right\rangle _{i}=\left|n\right\rangle _{j}$
	to denote the input to output transformation.
\end{proposition}

\begin{proof}
	From (\ref{eq:BSFock}), we have
	\begin{align}
	\left\langle 1\right|_{i}U_{ij}\left(\theta\right)\left|n\right\rangle _{i}\left|0\right\rangle _{j} & =\sqrt{n}\cos\theta\sin^{n-1}\theta\left|n-1\right\rangle _{j}\\
	& =\hat{a}_{j}\cos\theta\sin^{n-1}\theta\left|n\right\rangle _{j}\\
	& =\cos\theta\left(\sin\theta\right)^{\hat{n}_{ij}}\hat{a}_{j}\left|n\right\rangle _{j}
	\end{align}
	which implies, for a general state $\left|\psi\right\rangle _{i}=\sum_{n}c_{n}\left|n\right\rangle _{i},$
	the action of the channel is as given in (\ref{eq:N01channel}).
\end{proof}
\begin{remark}
	The displacement operation $\hat{D}_{i}\left(\lambda\right)$ on a
	mode $i$ can be implemented using a beam splitter $U_{ij}\left(\theta\right)$
	with the mode $j$ in the coherent state $\left|\lambda/\sin^{2}\theta\right\rangle _{j}.$ 
\end{remark}

\begin{proposition}
	\label{prop:MunroProj}Consider the measurement depicted in Fig.~\ref{Munroproj}.
	The projection implemented by the measurement scheme on modes $2$
	and $3$ is given by $\hat{F}_{23}^\dagger$, where
	\begin{align}
	\hat{F}_{23}=-\cos^{2}\theta\left(\sin^{2}\theta\right)^{\lambda^{2}}\left(\frac{\lambda^{2}}{2}\left|0\right\rangle _{2}\left|0\right\rangle _{3}+\frac{1}{4}\left|1\right\rangle _{2}\left|1\right\rangle _{3}\right),
	\end{align}
	where $\theta$ is related to the transmissivity of the photon subtraction
	beam splitters and $\lambda$ is the amplitude of the displacement
	stages.
\end{proposition}

\begin{proof}
	The measurement scheme depicted in Fig. 1 implements
	the projection $\hat{F}_{23}^\dagger$, which is
	\begin{align}
	&\left\langle 0,0\right|_{23}\hat{D}_{3}\left(-\lambda\right)\mathcal{N}_{3\rightarrow3''}^{01}\left(\theta\right)_{\theta\rightarrow0}\hat{D}_{3}\left(\lambda\right)\hat{D}_{2}\left(\lambda\right)\nonumber\\
	&\mathcal{N}_{2\rightarrow2''}^{01}\left(\theta\right)_{\theta\rightarrow0}\hat{D}_{2}\left(-\lambda\right)U_{23}^{\dagger}\left(\pi/4\right)\nonumber\\
	&\mathcal{N}_{2\rightarrow2'}^{11}\left(\pi/4\right)\mathcal{N}_{3\rightarrow3'}^{11}\left(\pi/4\right)U_{23}\left(\pi/4\right),\label{eq:Munroproj}
	\end{align}
	where for brevity of notation the primes on the output modes of individual
	elements in the transformation are suppressed ahead of the subsequent
	elements, and $\lambda\in\mathbf{R}^{+}$.
	
	From, (\ref{eq:N01channel}), we have $\left\langle 0\right|_{2}\hat{D}_{2}\left(\lambda\right)\left(\mathcal{N}_{2\rightarrow2''}^{01}\left(\theta\right)\right)\hat{D}_{2}\left(-\lambda\right)$
	\begin{align}
	&=\left\langle 0\right|_{2}\hat{D}_{2}\left(\lambda\right)\cos\theta\left(\sin\theta\right)^{\hat{n}_{2}}\hat{a}_{2}\hat{D}_{2}\left(-\lambda\right)\\
	& =\cos\theta\left(\sin\theta\right)^{\lambda^{2}}\left\langle -\lambda\right|_{2}\hat{a}_{2}\hat{D}_{2}\left(-\lambda\right)\\
	& =\cos\theta\left(\sin\theta\right)^{\lambda^{2}}\left\langle 0\right|_{2}\hat{D}_{2}\left(\lambda\right)\hat{a}_{2}\hat{D}_{2}\left(-\lambda\right)\\
	& =\cos\theta\left(\sin\theta\right)^{\lambda^{2}}\left\langle 0\right|_{2}\left(\hat{a}_{2}-\lambda\right)\\
	& =\cos\theta\left(\sin\theta\right)^{\lambda^{2}}\left(\left\langle 1\right|_{2}-\lambda\left\langle 0\right|_{2}\right).
	\end{align}
	Likewise,
	\begin{align}
	&\left\langle 0\right|_{3}\hat{D}_{3}\left(-\lambda\right)\left(\mathcal{N}_{3\rightarrow3''}^{01}\left(\theta\right)\right)\hat{D}_{3}\left(\lambda\right)\nonumber\\
	&=\cos\theta\left(\sin\theta\right)^{\lambda^{2}}\left(\left\langle 1\right|_{3}+\lambda\left\langle 0\right|_{3}\right).
	\end{align}
	
	From (\ref{eq:N11channel}), 
	
	\begin{align}
	\mathcal{N}^{11}\left(\pi/4\right)\left|n\right\rangle  & =-\left(\frac{1}{\sqrt{2}}\right)^{n+1}+n\left(\frac{1}{\sqrt{2}}\right)^{2}\left(\frac{1}{\sqrt{2}}\right)^{n-1}\left|n\right\rangle \\
	& =\frac{\left(n-1\right)}{\left(\sqrt{2}\right)^{n+1}}\left|n\right\rangle .
	\end{align}
	
	Therefore, we have 
	\[
	\mathcal{N}^{11}\left(\pi/4\right)=\frac{\left(\hat{n}-1\right)}{\left(\sqrt{2}\right)^{\hat{n}+1}}.
	\]
	
	Thus, $\hat{F}_{23}^\dagger$ in (\ref{eq:Munroproj}) can be written as
	\begin{widetext}
	\begin{align}
	\hat{F}_{23}^{\dagger} & =\cos^{2}\theta\left(\sin^{2}\theta\right)^{\lambda^{2}}\left(\left\langle 1\right|_{3}+\lambda\left\langle 0\right|_{3}\right)\left(\left\langle 1\right|_{2}-\lambda\left\langle 0\right|_{2}\right)U_{23}^{\dagger}\left(\pi/4\right)\frac{\left(\hat{n}_{2}-1\right)}{\left(\sqrt{2}\right)^{\hat{n}_{2}+1}}\frac{\left(\hat{n}_{3}-1\right)}{\left(\sqrt{2}\right)^{\hat{n}_{3}+1}}U_{23}\left(\pi/4\right)\\
	& =\cos^{2}\theta\left(\sin^{2}\theta\right)^{\lambda^{2}}\left(\left\langle 1\right|_{3}+\lambda\left\langle 0\right|_{3}\right)\left(\left\langle 1\right|_{2}-\lambda\left\langle 0\right|_{2}\right)U_{23}^{\dagger}\left(\pi/4\right)\left(\hat{n}_{2}-1\right)\frac{1}{\left(\sqrt{2}\right)^{\hat{n}_{2}+\hat{n}_{3}+2}}\left(\hat{n}_{3}-1\right)U_{23}\left(\pi/4\right)\\
	& =\cos^{2}\theta\left(\sin^{2}\theta\right)^{\lambda^{2}}\left(\left\langle 1\right|_{3}+\lambda\left\langle 0\right|_{3}\right)\left(\left\langle 1\right|_{2}-\lambda\left\langle 0\right|_{2}\right)\nonumber\\
	& \times U_{23}^{\dagger}\left(\pi/4\right)\left(\hat{n}_{2}-1\right)U_{23}\left(\pi/4\right)U_{23}^{\dagger}\left(\pi/4\right)\frac{1}{\left(\sqrt{2}\right)^{\hat{n}_{2}+\hat{n}_{3}+2}}U_{23}\left(\pi/4\right)U_{23}^{\dagger}\left(\pi/4\right)\left(\hat{n}_{3}-1\right)U_{23}\left(\pi/4\right).
	\end{align}
	\end{widetext}
	From (\ref{eq:BSinputoutput}), we have 
	\begin{align}
	U_{23}^{\dagger}\left(\pi/4\right)\left(\hat{n}_{3}-1\right)U_{23}\left(\pi/4\right) & =\frac{\left(\hat{a}_{2}^{\dagger}-\hat{a}_{3}^{\dagger}\right)\left(\hat{a}_{2}-\hat{a}_{3}\right)}{2}-1\\
	U_{23}^{\dagger}\left(\pi/4\right)\left(\hat{n}_{2}-1\right)U_{23}\left(\pi/4\right) & =\frac{\left(\hat{a}_{2}^{\dagger}+\hat{a}_{3}^{\dagger}\right)\left(\hat{a}_{2}+\hat{a}_{3}\right)}{2}-1.
	\end{align}
	Thus, we have
	\begin{align}
	\hat{F}_{23}^{\dagger}&=\cos^{2}\theta\left(\sin^{2}\theta\right)^{\lambda^{2}}\left(\left\langle 1\right|_{3}+\lambda\left\langle 0\right|_{3}\right)\left(\left\langle 1\right|_{2}-\lambda\left\langle 0\right|_{2}\right)\nonumber\\
	&\times\left(\frac{\left(\hat{a}_{2}^{\dagger}+\hat{a}_{3}^{\dagger}\right)\left(\hat{a}_{2}+\hat{a}_{3}\right)}{2}-1\right)\nonumber\\
	&\times\frac{1}{\left(\sqrt{2}\right)^{\hat{n}_{3}+\hat{n}_{2}+2}}\left(\frac{\left(\hat{a}_{2}^{\dagger}-\hat{a}_{3}^{\dagger}\right)\left(\hat{a}_{2}-\hat{a}_{3}\right)}{2}-1\right).
	\end{align}
	The projector can be written in the ket form as $\hat{F}_{23}$
	\begin{widetext}
	\begin{align}
	& =\cos^{2}\theta\left(\sin^{2}\theta\right)^{\lambda^{2}}\left(\frac{\left(\hat{a}_{2}^{\dagger}+\hat{a}_{3}^{\dagger}\right)\left(\hat{a}_{2}+\hat{a}_{3}\right)}{2}-1\right)\frac{1}{\left(\sqrt{2}\right)^{\hat{n}_{3}+\hat{n}_{2}+2}}\left(\frac{\left(\hat{a}_{2}^{\dagger}-\hat{a}_{3}^{\dagger}\right)\left(\hat{a}_{2}-\hat{a}_{3}\right)}{2}-1\right)\left(\left|1\right\rangle _{2}-\lambda\left|0\right\rangle _{2}\right)\left(\left|1\right\rangle _{3}+\lambda\left|0\right\rangle _{3}\right)\nonumber\\
	& =\cos^{2}\theta\left(\sin^{2}\theta\right)^{\lambda^{2}}\left(\frac{\left(\hat{a}_{2}^{\dagger}+\hat{a}_{3}^{\dagger}\right)\left(\hat{a}_{2}+\hat{a}_{3}\right)-2}{8}\right)\left(-\frac{1}{\sqrt{2}}\left(\left|2\right\rangle _{2}\left|0\right\rangle _{3}+\left|0\right\rangle _{2}\left|2\right\rangle _{3}\right)+2\lambda^{2}\left|0\right\rangle _{2}\left|0\right\rangle _{3}\right)\nonumber\\
	& =-\cos^{2}\theta\left(\sin^{2}\theta\right)^{\lambda^{2}}\left(\frac{\lambda^{2}}{2}\left|0\right\rangle _{2}\left|0\right\rangle _{3}+\frac{1}{4}\left|1\right\rangle _{2}\left|1\right\rangle _{3}\right).
	\end{align}
	
	\end{widetext}
	
\end{proof}
%

\begin{remark}
	$\hat{F}_{23}$ can be expressed in a weighted
	normalized form as
	\[
	\left(\frac{-\cos^{2}\theta\left(\sin^{2}\theta\right)^{\lambda^{2}}\sqrt{1+4\lambda^{4}}}{4}\right)\left(\frac{2\lambda^{2}\left|0\right\rangle _{2}\left|0\right\rangle _{3}+\left|1\right\rangle \left|1\right\rangle _{3}}{\sqrt{1+4\lambda^{4}}}\right).
	\]
\end{remark}

\begin{remark}
	Thus, states
	\begin{align*}
	\left|\psi\right\rangle _{12} & =\frac{\left|0\right\rangle _{1}\left|0\right\rangle _{2}+\xi\left|1\right\rangle _{1}\left|1\right\rangle _{2}}{\sqrt{1+\xi^{2}}},\\
	\left|\psi\right\rangle _{43} & =\frac{\left|0\right\rangle _{4}\left|0\right\rangle _{3}+\xi\left|1\right\rangle _{4}\left|1\right\rangle _{3}}{\sqrt{1+\xi^{2}}},
	\end{align*}
	can be entanglement swapped into a state $\left|\psi\right\rangle _{14}=\left(\left|0\right\rangle _{1}\left|0\right\rangle _{4}+\xi\left|1\right\rangle _{1}\left|1\right\rangle _{4}\right)/\sqrt{1+\xi^{2}},$
	using $\hat{F}_{23}$ of Prop.~\ref{prop:MunroProj} with $\lambda=\sqrt{\xi/2}.$
	The success probability of physically implementing the projection $\hat{F}_{23}$ on these states is given by
	\begin{align}
	P_{\textrm{phys}}=\left(\frac{\xi^{2}}{1+\xi^{2}}\right)\frac{\cos^{4}\theta\left(\sin^{2}\theta\right)^{\xi}}{16},\label{pphys}
	\end{align}
	which can be numerically optimized over the reflectivity parameter $\theta$.
\end{remark}

\section{Iterative formula for a chain of repeater links connected by the non-Gaussian entanglement swap \label{iterform_app}}
Consider the proposed CV repeater link with a single quantum scissors. As an alternative to the Fock basis description of Appendix \ref{rci_appen}), the state successfully heralded across the repeater link can be expressed as~\cite{BASRL14}
\begin{align}
\ket{\psi^{(1)}}_{A_1B_1L_1} &=\frac{1}{\gamma^{(1)}} \left(\gamma_0^{(1)} + \gamma_1^{(1)}a_1^\dag + \gamma_2^{(1)}b_1^\dag + \gamma_3^{(1)}a_1^\dag b_1^\dag\right)\nonumber\\
&\times\sigma_{A_1L_1}^{\rho}\ket{0}_{A_1B_1L_1}\,,
\end{align}
where
\begin{align}
&\gamma^{(1)}_0=f,\\
&\gamma^{(1)}_1=0,\\
&\gamma^{(1)}_2=0,\\
&\gamma^{(1)}_3=\kappa_h f\\
&f=\frac{\sqrt{\kappa}\sech r}{\sqrt{\sech^2r+t \tanh^2r}},\\
&\gamma^{(1)}=\sqrt{\frac{\kappa \sech^2r+t\tanh^2r}{\cosh^2 r(\sech^2r+t \tanh^2r)^2}},\\
&\kappa_h=\sqrt{\frac{1-\kappa}{\kappa}}\sqrt{t}\tanh r,\\
&\tanh \rho=\sqrt{1-t}\tanh r,
\end{align}
$\sigma_{AL}^{\rho}$ is the two-mode squeezing operator corresponding to squeezing of magnitude $\rho$ in modes $A_1,\ L_1$, and $r=\left(\sinh^{-1}\sqrt{\mu}\right)$. The heralding probability is given by
\be
P_{\textrm{sciss}}={\gamma^{(1)}}^2.
\ee

The state obtained by connecting two such repeater links using a Bell swap projection $\hat{\Pi}_{B_1A_2}=|\phi\rangle\langle\phi|_{B_1A_2}$, where
\be
\ket{\phi}_{B_1A_2}=\frac{1}{\sqrt{1+q^2}}(\ket{00}_{B_1A_2}+q\ket{11}_{B_1A_2}), \ q\in\mathbbm{R}^+\,,
\ee 
is given by
\begin{align}
&\bra{\phi}_{B_1A_2}\ket{\psi^{(1)}}_{A_1B_1L_1}\otimes\ket{\psi^{(1)}}_{A_2B_2L_2}\nonumber\\
&=\left(\frac{f}{\gamma^{(1)}}\right)^2\bra{\phi}_{B_1A_2}(1+\kappa a_1^\dag b_1^\dag)(1+\kappa a_2^\dag b_2^\dag)\nonumber\\
&\times\sigma^{\rho}_{A_1L_1}\sigma^{\rho}_{A_2L_2}\ket{0}_{A_1B_1L_1A_2B_2L_2} \label{unnormEL}.
\end{align}
As a result, the state of the modes $A_1, \ B_2, \ L_1$ (with the mode $L_2$ being traced over) is heralded as
\begin{align}
\ket{\psi^{(2)}}_{A_1B_2L_1} &=\frac{1}{\gamma^{(2)}}  \left(\gamma_0^{(2)} + \gamma_1^{(2)}a_1^\dag + \gamma_2^{(2)}b_2^\dag + \gamma_3^{(2)}a_1^\dag b_2^\dag\right)\nonumber\\
&\times\sigma_{A_1L_1}^{\rho}\ket{0}_{A_1B_2L_1}\,,
\end{align}
where
\begin{align}
&\gamma^{(2)}_0=1,\\
&\gamma^{(2)}_1=\kappa_h q \tanh \rho\\
&\gamma^{(2)}_2=0\\
&\gamma^{(2)}_3=\kappa_h^2 q\\
&\gamma^{(2)}=\sqrt{1+\kappa_h^2 q^2 \sinh^2 \rho+\kappa_h^4 q^2 \cosh^2 \rho}.
\end{align}
The corresponding success probability is given by the product of the ideal Bell swap projection probability for the repeater link states $P_{\hat{\Pi}}$ times the probability of physically implementing the projection using linear optics $P_{\textrm{phys}}$ of (\ref{pphys}), i.e.,
\begin{align}
P_{\textrm{swap}}=P_{\hat{\Pi}}\times P_{\textrm{phys}},
\end{align}
where the former is the norm of the unnormalized state in~(\ref{unnormEL}), given by
\begin{align}
P_{\hat{\Pi}}=\frac{f^4\sech^2\rho}{(1+q^2){\gamma^{(1)}}^4}{\gamma^{(2)}}^2
\end{align}

Now, say we want to concatenate two such states $\ket{\psi^{(2)}}_{ABL}$ (connected by the non-Gaussian Bell state projection), to obtain the state $\ket{\psi^{(3)}}_{ABL}$ across four repeater links, or similarly concatenate two states $\ket{\psi^{(3)}}_{ABL}$ to obtain the state across eight repeater links $\ket{\psi^{(4)}}_{ABL}$. More generally, assume that at the $i^{th}$ step of concatenation, we have two states whose tensor product is
\begin{align}
&\ket{\psi^{(i)}}_{A_1B_1L_1}\otimes \ket{\psi^{(i)}}_{A_2B_2L_2} \nonumber\\
&= \left(\gamma_0^{(i)} + \gamma_1^{(i)}a_1^\dag + \gamma_2^{(i)}b_1^\dag + \gamma_3^{(i)}a_1^\dag b_1^\dag\right)\nonumber\\
&\times\left(\gamma_0^{(i)} + \gamma_1^{(i)}a_2^\dag + \gamma_2^{(i)}b_2^\dag + \gamma_3^{(i)}a_2^\dag b_2^\dag\right)\nonumber\\
&\times\sigma_{A_1L_1}^{\rho}\sigma_{A_2L_2}^{\rho}\ket{0}_{A_1B_1L_1A_2B_2L_2}\,,
\end{align}
where for brevity of notation, we have denoted the modes as $A_1,\ B_1,\ L_1,\ A_2, B_2,\ L_2$ in place of the actual mode labels. When the modes $B_1$ and $A_2$ are projected on the non-Gaussian Bell state, we have
\begin{align}
&\bra{\phi}_{B_1A_2}\ket{\psi^{(i)}}_{A_1B_1L_1}\otimes \ket{\psi^{(i)}}_{A_2B_2L_2}\nonumber\\
&=\left(a\bra{00}_{B_1A_2} + b\bra{01}_{B_1A_2} +c\bra{10}_{B_1A_2} + d\bra{11}_{B_1A_2}\right)\nonumber\\
&\times\sigma_{A_1L_1}^{\rho}\sigma_{A_2L_2}^{\rho}\ket{0}_{A_1B_1L_1A_2B_2L_2}\,,
\end{align}
where
\begin{align}
a&=\frac{(\gamma_0^{(i)}+\gamma_1^{(i)}a_1^\dag)(\gamma_0^{(i)}+\gamma_2^{(i)}b_2^\dag)}{\sqrt{1+q^2}}\nonumber\\ &+\frac{(\gamma_2^{(i)}+\gamma_3^{(i)}a_1^\dag)(\gamma_1^{(i)}+\gamma_3^{(i)}b_2^\dag)q}{\sqrt{1+q^2}} \\
b&=\frac{(\gamma_2^{(i)}+\gamma_3^{(i)}a_1^\dag)(\gamma_0^{(i)}+\gamma_2^{(i)}b_2^\dag)q}{\sqrt{1+q^2}}\\
c&=\frac{(\gamma_0^{(i)}+\gamma_1^{(i)}a_1^\dag)(\gamma_1^{(i)}+\gamma_3^{(i)}b_2^\dag)q}{\sqrt{1+q^2}}\\
d&=\frac{(\gamma_0^{(i)}+\gamma_1^{(i)}a_1^\dag)(\gamma_0^{(i)}+\gamma_2^{(i)}b_2^\dag)q}{\sqrt{1+q^2}}.
\end{align}
The resulting state that is heralded in modes $A_1,\ B_2,\ L_1$ is given by
\begin{align}
&\ket{\psi^{(i+1)}}_{A_1B_2L_1}\nonumber\\
&=\frac{1}{\gamma^{(i+1)}} \left(\gamma_0^{(i+1)} + \gamma_1^{(i+1)}a_1^\dag + \gamma_2^{(i+1)}b_2^\dag+\gamma_2^{(i+1)}a_1^\dag b_2^\dag\right)\nonumber\\
&\times\sigma_{A_1L_1}^{\rho}\ket{0}\,,
\end{align}
where
\begin{align}
&\gamma_0^{(i+1)}=\left({\gamma^{(i)}_0}^2 + q\gamma^{(i)}_2(\gamma^{(i)}_1+\gamma^{(i)}_0\tanh p)\right)\\
&\gamma_1^{(i+1)}=\left(\gamma^{(i)}_0\gamma^{(i)}_1+q\gamma^{(i)}_3(\gamma^{(i)}_1+\gamma^{(i)}_0\tanh \rho)\right)\\
&\gamma_2^{(i+1)}=\left(\gamma^{(i)}_0\gamma^{(i)}_2+q\gamma^{(i)}_2(\gamma^{(i)}_3+\gamma^{(i)}_2\tanh \rho)\right)\\
&\gamma_3^{(i+1)}=\left(\gamma^{(i)}_1\gamma^{(i)}_2 + q\gamma^{(i)}_3(\gamma^{(i)}_3+\gamma^{(i)}_2\tanh \rho)\right),\\
&\gamma^{(i+1)}\nonumber\\
&=\sqrt{\left({\gamma_0^{(i+1)}}^2+{\gamma_2^{(i+1)}}^2\right)+\cosh^2\rho\left({\gamma_1^{(i+1)}}^2+{\gamma_3^{(i+1)}}^2\right)}.
\end{align}
The state can be simplified as
\begin{widetext}
\begin{align}
\ket{\psi^{(i+1)}}_{A_1B_2L_1}=\frac{1}{\gamma^{(i+1)}}\left[\left(\gamma_0^{(i+1)} + \gamma_1^{(i+1)}a_1^\dag\right)\sigma_{A_1L_1}^{\rho}\ket{0}_{A_1L_1}\otimes\ket{0}_{B_2}+\left(\gamma_2^{(i+1)} + \gamma_3^{(i+1)}a_1^\dag\right)\sigma_{A_1L_1}^{\rho}\ket{0}_{A_1L_1}\otimes\ket{1}_{B_2}\right]
\end{align}
\end{widetext}

The end-to-end two-mode state heralded across a repeater chain of $n=2^x,\ x\in\mathbb{N},$ repeater links can be written down by tracing over the environment mode as
\begin{widetext}
\begin{align}
\hat{\rho}_{A_1B_{2^x}} = &\operatorname{Tr}_{L_1}\left(|\psi^{(x+1)}\rangle\langle\psi^{(x+1)}|_{A_1B_{2^x}L_1}\right)\nonumber\\
&=\hat{\rho}_{A_1}^{(0,0)}\otimes |0\rangle\langle 0|_{B_{2^x}}+\hat{\rho}_{A_1}^{(0,1)}\otimes |0\rangle\langle 1|_{B_{2^x}}+\hat{\rho}_{A_1}^{(1,0)}\otimes |1\rangle\langle 0|_{B_{2^x}}+\hat{\rho}_{A_1}^{(1,1)}\otimes |1\rangle\langle 1|_{B_{2^x}},
\end{align}
where
\begin{align}
&\hat{\rho}_{A_1}^{(0,0)}=\frac{1}{{\gamma^{(x+1)}}^2}\left(\gamma_0^{(x+1)} + \gamma_1^{(x+1)}a_1^\dag\right)\hat{\rho}^{th}_{A_1}(\rho)\left(\gamma_0^{(x+1)} + \gamma_1^{(x+1)}a_1\right)\nonumber\\
&\hat{\rho}_{A_1}^{(0,1)}=\frac{1}{{\gamma^{(x+1)}}^2}\left(\gamma_0^{(x+1)} + \gamma_1^{(x+1)}a_1^\dag\right)\hat{\rho}^{th}_{A_1}(\rho)\left(\gamma_2^{(x+1)} + \gamma_3^{(x+1)}a_1\right)\nonumber\\
&\hat{\rho}_{A_1}^{(1,0)}=\frac{1}{{\gamma^{(x+1)}}^2}\left(\gamma_2^{(x+1)} + \gamma_3^{(x+1)}a_1^\dag\right)\hat{\rho}^{th}_{A_1}(\rho)\left(\gamma_0^{(x+1)} + \gamma_1^{(x+1)}a_1\right)\nonumber\\
&\hat{\rho}_{A_1}^{(1,1)}=\frac{1}{{\gamma^{(x+1)}}^2}\left(\gamma_2^{(x+1)} + \gamma_3^{(x+1)}a_1^\dag\right)\hat{\rho}^{th}_{A_1}(\rho)\left(\gamma_2^{(x+1)} + \gamma_3^{(x+1)}a_1\right),
\end{align}
\end{widetext}
where $\hat{\rho}^{th}$ is the thermal state of mean photon number $\rho$. The above density operator can be written in the Fock basis with terms $\langle m_1, m_2|_{A_1B_{2^x}}.\hat{\rho}_{A_1B_{2^x}}.|n_1,n_2\rangle_{A_1B_{2^x}}=\hat{\rho}_{A_1}^{(m_2,n_2)},\ \{m_1,n_1\in\mathbb{W}\}$ and $\{m_2,n_2\in\{0,1\}\},$ where

\begin{align}
&\hat{\rho}_{A_1}^{(0,0)}=\langle \phi^{(0)}_{m_1}|\hat{\rho}^{th}_{A_1}(\rho)|\phi^{(0)}_{n_1}\rangle\nonumber\\
&\hat{\rho}_{A_1}^{(0,1)}=\langle \phi^{(0)}_{m_1}|\hat{\rho}^{th}_{A_1}(\rho)|\phi^{(1)}_{n_1}\rangle\nonumber\\
&\hat{\rho}_{A_1}^{(1,0)}=\langle \phi^{(1)}_{m_1}|\hat{\rho}^{th}_{A_1}(\rho)|\phi^{(0)}_{n_1}\rangle\nonumber\\
&\hat{\rho}_{A_1}^{(1,1)}=\langle \phi^{(1)}_{m_1}|\hat{\rho}^{th}_{A_1}(\rho)|\phi^{(1)}_{n_1}\rangle,
\end{align}
with
\begin{align}
|\phi^{(0)}_{n}\rangle=\left(\gamma_0^{(x+1)}|n\rangle+\gamma_1^{(x+1)}\sqrt{n}|n-1\rangle\right)/\gamma^{(x+1)},\nonumber\\
|\phi^{(1)}_{n}\rangle=\left(\gamma_2^{(x+1)}|n\rangle+\gamma_3^{(x+1)}\sqrt{n}|n-1\rangle\right)/\gamma^{(x+1)}.
\end{align}

Finally, the RCI of the state can thus be calculated from the eigenspectra of the suitably truncated Fock basis density matrices corresponding to $\hat{\rho}_{A_1B_{2^x}}$ and $\hat{\rho}_{A_1}$.

\bibliographystyle{apsrev4-1.bst}
\bibliography{ref2}

\end{document}